\def\be{\begin{equation}}
\def\ee{\end{equation}}
\def\ba{\begin{array}}
\def\ea{\end{array}}

\documentclass[twocolumn,
showpacs,preprintnumbers,amsmath,amssymb]{revtex4}
\usepackage{multirow}
\usepackage{times,epsfig,amsthm,amssymb,amsfonts,amsmath,graphicx}
\def\qed{\leavevmode\unskip\penalty9999 \hbox{}\nobreak\hfill
    \quad\hbox{\leavevmode  \hbox to.77778em{%
               \hfil\vrule   \vbox to.675em%
               {\hrule width.6em\vfil\hrule}\vrule\hfil}}
     \par\vskip3pt}

\newtheorem{theorem}{Theorem}

\newtheorem{lemma}{Lemma}
\newcommand{\ket}[1]{|#1\rangle}
\newcommand{\bra}[1]{\langle #1|}

\newcommand{\beq}{\begin{equation}}
\newcommand{\eeq}{\end{equation}}
\newcommand{\beqa}{\begin{eqnarray}}
\newcommand{\eeqa}{\end{eqnarray}}

\begin{document}

\title{ General Monogamy Relations for Multiqubit W-class States in terms of convex-roof extended negativity of assistance and squared R\'{e}nyi-$\alpha$ entanglement}

\author{Yanying Liang$^{1}$, Xiufang Feng$^{1}$ and Wei Chen$^{2}$}

\affiliation{$^{1}$ School of Mathematics, South China University of Technology, Guangzhou 510641, China\\
$^{2}$School of Computer Science and Network Security, Dongguan University of Technology, Dongguan 523808, China}

\begin{abstract}

For multipartite entangled states, entanglement monogamy is an important property. We present some new analytical monogamy inequalities satisfied by the $x$-th power of the dual of convex-roof extended negativity (CREN), namely CREN of Assistance (CRENoA), with $x\geq2$ and $x\leq0$ for multiqubit generalized W-class states. We also provide the upper bound of the squared R\'{e}nyi-$\alpha$ entanglement (SR$\alpha$E) with $\alpha$ in the region $[(\sqrt 7  - 1)/2,(\sqrt {13}  - 1)/2]$ for multiqubit generalized W-class states.
\end{abstract}


\maketitle

\section{Introduction}

While classical correlation can be freely shared among parties in
multi-party systems, quantum entanglement is restricted in its
shareability. If a pair of parties are maximally entangled in
multipartite systems, they cannot have any entanglement \cite{CKW,OV}
nor classical correlations \cite{KW} with the rest of the system.
This restriction of entanglement shareability among multi-party systems is known as
the monogamy of entanglement (MoE) \cite{t1, t2, t3, t4, t5, t6}.

The monogamy of entanglement (MoE) is one of the fundamental differences between quantum entanglement and classical correlations that a quantum system entangled with one of the other systems limits its entanglement with the remaining others. For example, MoE is a key ingredient to make quantum
cryptography secure because it quantifies how much information
an eavesdropper could potentially obtain about the secret key to be
extracted \cite{JMR}.

Coffman, Kundu, and Wootters established the first quantitative characterization of the MoE for the squared concurrence (SC) \cite{WKW,TJO, YKB, ZXN1, JSK1} in an arbitrary three-qubit quantum state.
Another two well-known entanglement measures are convex-roof extended negativity (CREN) \cite{SL} and R\'{e}nyi-$\alpha$ entanglement (R$\alpha$E) \cite{RH}. CREN is a good alternative for MoE without any known example violating its property even in higher-dimensional systems and R$\alpha$E is the generalization of entanglement of formation. Recently, the general monogamy relations for the $x$-th power of CREN has been shown for a mixed state $\rho_{A_1A_2\ldots A_N}$ in a $N$-qubit system \cite{YL},
\begin{equation}\
\mathcal{\widetilde{N}}^x_{A_1|A_2\ldots A_N}\geq \mathcal{\widetilde{N}}^x_{A_1A_2}+...+\mathcal{\widetilde{N}}^x_{A_1A_{N}},
\end{equation}
for $x\geq2$
and
\begin{equation}
\mathcal{\widetilde{N}}^x_{A_1|A_2\ldots A_N}< \mathcal{\widetilde{N}}^x_{A_1A_2}+...+\mathcal{\widetilde{N}}^x_{A_1A_{N}},  \end{equation}
for $x\leq0$.
Two years ago, Wei Song and Yan-Kui Bai showed the properties of the squared R\'{e}nyi-$\alpha$ entanglement (SR$\alpha$E) and proved that the lower bound of SR$\alpha$E in an arbitrary $N$-qubit mixed state \cite{WS},
\begin{equation}\
E_\alpha^2(\rho_{A_1|A_2\dots A_n})\geq
E_\alpha^2(\rho_{A_1A_2})+...+E_\alpha^2(\rho_{A_1A_n}),
\end{equation}
where $E^2_\alpha(\rho_{A_1|A_2\dots A_n})$ quantifies the entanglement in the partition
$A_1|A_2\dots A_n$ and $E_\alpha^2(\rho_{A_1A_i})$ quantifies the one in two-qubit subsystem $A_1A_i$
with the order $\alpha\ge (\sqrt 7  - 1)/2$.

In this paper, we show the general monogamy relations for the $x$-th power of CRENoA of generalized multiqubit W-class states. This part provides a more efficient way for MoE. We also prove that the SR$\alpha$E with the order $\alpha$ ranges in the region $[(\sqrt 7  - 1)/2, (\sqrt {13}  - 1)/2]$ also obeys a general monogamy relation for arbitrary generalized multiqubit W-class states.

\section{Monogamy of concurrence and convex-roof extended negativity}

Given a bipartite pure state $|\psi\rangle_{AB}$ in a $d \otimes d^{'} (d\leq d^{'})$ quantum system,  its concurrence,  $C(|\psi\rangle_{AB})$ is defined as \cite{PR}
\begin{equation}\label{Concur}
\mathcal{C}(|\psi\rangle_{AB}) = \sqrt{2[1-Tr(\rho_{A}^{2})]},
\end{equation}
where $\rho_{A}$ is reduced density matrix by tracing over the subsystem $B, $ $\rho_{A} = Tr_{B}(|\psi\rangle_{AB}\langle\psi|)$ (and analogously for $\rho_{B}$). For any mixed state $\rho_{AB}, $ the concurrence is given by the minimum average concurrence taken over all decompositions of $\rho_{AB}, $
the so-called convex roof
\begin{equation}
\mathcal{C}(\rho_{AB}) = \min_{\{p_{i}, |\psi_{i}\rangle\}}\sum_{i}p_{i}\mathcal{C}(|\psi_{i}\rangle),
\end{equation}
where the convex roof is notoriously hard to evaluate and therefore it is difficult to determine whether or not an arbitrary state is entangled.

Similarly,  the concurrence of assistance (CoA) of any mixed state $\rho_{AB}$ is defined as \cite{CSY}
\begin{equation}
\mathcal{C}_{a}(\rho_{AB}) = \max_{\{p_{i}, |\psi_{i}\rangle\}}\sum_{i}p_{i}\mathcal{C}(|\psi_{i}\rangle),
\end{equation}
where the maximum is taken over all possible pure state decompositions $\{p_{i}, |\psi_{i}\rangle\}$ of $\rho_{AB}$.

Another well-known quantification of bipartite entanglement is convex-roof extended negativity (CREN).
For a bipartite mixed state $\rho_{AB}, $ CREN is defined as
\begin{equation}
\widetilde{\mathcal{N}}(\rho_{AB}) = \min_{\{p_{i}, |\psi_{i}\rangle\}}\sum_{i}p_{i}\mathcal{N}(|\psi_{i}\rangle),
\end{equation}
where the minimum is taken over all possible pure state decompositions $\{p_{i}, |\psi_{i}\rangle\}$ of $\rho_{AB}$.

Similar to the duality between concurrence and CoA,  we can also define a dual to CREN,  namely CRENoA,  by taking the maximum value of average negativity over all
possible pure state decomposition,  i.e.
\begin{equation}
\widetilde{\mathcal{N}}_{a}(\rho_{AB}) = \max_{\{p_{i}, |\psi_{i}\rangle\}}\sum_{i}p_{i}\mathcal{N}(|\psi_{i}\rangle),
\end{equation}
where the maximum is taken over all possible pure state decompositions $\{p_{i}, |\psi_{i}\rangle\}$ of $\rho_{AB}$.

In the following we study the monogamy property of the CRENoA for the $n$-qubit
generalized W-class states $|\psi\rangle\in H_{A_1}\otimes H_{A_2}\otimes...\otimes H_{A_n}$ defined by
\begin{equation}
|\psi\rangle=a|000...\rangle+b_1|01...0\rangle+...+b_n|00...1\rangle,
\end{equation}
with $|a|^2+\sum_{i=1}^{n}|b_i|=1.$

\begin{lemma}\label{c=ca}
For $n$-qubit generalized W-class states (9),   we have
\begin{eqnarray}\label{cca}
\widetilde{\mathcal{N}}(\rho_{A_1A_i})=\widetilde{\mathcal{N}}_{a}(\rho_{A_1A_i}),
\end{eqnarray}
where $\rho_{A_1A_i}=Tr_{A_2...A_{i-1}A_{i+1}...A_{n}}(|\psi\rangle\langle\psi|)$.
\end{lemma}

\begin{proof}
We assume
$\rho_{A_1A_i}=|x\rangle_{A_1A_i}\langle x|+|y\rangle_{A_1A_i}\langle y|$ \cite{JSK2},
where
$$
\ba{l}
|x\rangle_{A_1A_{i}}=a|00\rangle_{A_1A_{i}} +b_1|10\rangle_{A_1A_{i}}+b_i|01\rangle_{A_1A_{i}},  \\[2mm]
|y\rangle_{A_1A_{i}}=\sqrt{\sum_{k\neq i}|b_k|^2}|00\rangle_{A_1A_{i}}.
\ea
$$

From the HJW theorem in Ref. \cite{JSK2},  for any pure-state decomposition of
$\rho_{A_1A_{i}}=\sum_{h=1}^{r}|\phi_h\rangle_{A_1A_{i}}\langle\phi_h|$,    one has
$|\phi_h\rangle_{A_1A_{i}} =u_{h1}|x\rangle_{A_1A_{i}}+u_{h2}|y\rangle_{A_1A_{i}}$ for
some $r\times r$ unitary matrices $u_{h1}$ and $u_{h2}$ for each $h$.
Consider the normalized bipartite pure state $|\tilde{\phi_h}\rangle_{A_1A_{i}}=|\phi_h\rangle_{A_1A_{i}}/\sqrt{p_h}$ with $p_h=|\langle\phi_h|\phi_h\rangle|$.
In Ref. \cite{JSK3},    for any bipartite pure state $|\psi\rangle$,    one has
\begin{equation}\nonumber
C(|\psi\rangle)=\mathcal{N}(|\psi\rangle).
\end{equation}
and combining with the Lemma 1 in Ref. \cite{ZXN2}, for $|\tilde{\phi_h}\rangle_{A_1A_{i}}$,   we have
\begin{equation}\nonumber
\mathcal{N}(|\tilde{\phi_h}\rangle_{A_1A_{i}})=\frac{2}{p_h}|u_{hi}|^2|b_1||b_i|.
\end{equation}

Then combing (7) and (8),  we can obtain
\begin{eqnarray}\nonumber
\widetilde{\mathcal{N}}(\rho_{A_1A_i})
&=&
\min_{\{p_h, |\tilde{\phi_h}\rangle_{A_1A_{i}}\}}
\sum_hp_h\mathcal{N}(|\tilde{\phi_h}\rangle_{A_1A_{i}})\\[1mm]\nonumber
&=&\max_{\{p_h, |\tilde{\phi_h}\rangle_{A_1A_{i}}\}}
\sum_hp_h\mathcal{N}(|\tilde{\phi_h}\rangle_{A_1A_{i}})\\[1mm]\nonumber
&=&\widetilde{\mathcal{N}}_{a}(\rho_{A_1A_i}).
\end{eqnarray}
\end{proof}

\begin{theorem}\label{TH2}
For the $n$-qubit generalized W-class states $|\psi\rangle\in H_{A_1}\otimes H_{A_2}\otimes...\otimes H_{A_n}$,    the
CRENoA satisfies
\begin{eqnarray}\label{cax}
\widetilde{\mathcal{N}}_{a}^x(\rho_{A_1|A_{j_1}...A_{j_{m-1}}})\geq\frac{x}{2^x-1}\sum_{i=1}^{{m-1}}\widetilde{\mathcal{N}}_{a}^x(\rho_{A_1A_{j_i}}), \end{eqnarray}
where $x\geq2$ and $\rho_{A_1A_{j_1}...A_{j_{m-1}}}$ is the $m$-qubit,    $2\leq m\leq n$,    reduced density matrix of $|\psi\rangle$.
\end{theorem}

\begin{proof}
For the $n$-qubit generalized W-class state $|\psi\rangle$,
according to the definitions of $\widetilde{\mathcal{N}}(\rho)$ and $\widetilde{\mathcal{N}}_{a}(\rho)$,    one has
$\widetilde{\mathcal{N}}_{a}(\rho_{A_1|A_{j_1}...A_{j_{m-1}}})\geq \widetilde{\mathcal{N}}(\rho_{A_1|A_{j_1}...A_{j_{m-1}}})$. When $x\geq2$,    we have
\begin{widetext}
\begin{eqnarray}\nonumber
\widetilde{\mathcal{N}}_{a}^x(\rho_{A_1|A_{j_1}...A_{j_{m-1}}})
&\geq&\left(\frac{\widetilde{\mathcal{N}}_{a}(\rho_{A_1|A_{j_1}...A_{j_{m-1}}})+\widetilde{\mathcal{N}}(\rho_{A_1|A_{j_1}...A_{j_{m-1}}})}{2}\right)^x\\[1mm]\nonumber
&=&\frac{1}{2^x}\widetilde{\mathcal{N}}_{a}^x(\rho_{A_1|A_{j_1}...A_{j_{m-1}}})\left(1+\frac{\widetilde{\mathcal{N}}(\rho_{A_1|A_{j_1}...A_{j_{m-1}}})}{\widetilde{\mathcal{N}}_{a}(\rho_{A_1|A_{j_1}...A_{j_{m-1}}})}\right)^x\\[1mm]\nonumber
&\geq&\frac{1}{2^x}\widetilde{\mathcal{N}}_{a}^x(\rho_{A_1|A_{j_1}...A_{j_{m-1}}})\left(1+x\frac{\widetilde{\mathcal{N}}(\rho_{A_1|A_{j_1}...A_{j_{m-1}}})}{\widetilde{\mathcal{N}}_{a}(\rho_{A_1|A_{j_1}...A_{j_{m-1}}})}\right)^x\\[1mm]\nonumber
&=&\frac{1}{2^x}\widetilde{\mathcal{N}}_{a}^x(\rho_{A_1|A_{j_1}...A_{j_{m-1}}})+\frac{x}{2^x}\widetilde{\mathcal{N}}^x(\rho_{A_1|A_{j_1}...A_{j_{m-1}}})\\[1mm]\nonumber
\end{eqnarray}
\end{widetext}
Here we have used in the first inequality the inequality $a^x\geq(\frac{a+b}{2})^x$ for $a\geq b>0$ and $x\geq0$.
The second inequality is due to $(1+t)^x\geq 1+xt^x$ for $x\geq1$ and $1\geq t\geq0$.

Then we have
\begin{eqnarray}\nonumber
\widetilde{\mathcal{N}}_{a}^x(\rho_{A_1|A_{j_1}...A_{j_{m-1}}})
&\geq&\frac{x}{2^x-1}\widetilde{\mathcal{N}}^x(\rho_{A_1|A_{j_1}...A_{j_{m-1}}})
\end{eqnarray}

Combining with Lemma1,  we have
\begin{eqnarray}\nonumber
\widetilde{\mathcal{N}}_{a}^x(\rho_{A_1|A_{j_1}...A_{j_{m-1}}})
&\geq&\frac{x}{2^x-1}\widetilde{\mathcal{N}}^x(\rho_{A_1|A_{j_1}...A_{j_{m-1}}})\\[1mm]\nonumber
&\geq&\frac{x}{2^x-1}\sum_{i=1}^{m-1}\widetilde{\mathcal{N}}^x(\rho_{A_1A_{j_i}})\\[1mm]\nonumber
&=&\frac{x}{2^x-1}\sum_{i=1}^{m-1}\widetilde{\mathcal{N}}^x_{a}(\rho_{A_1A_{j_i}}).
\end{eqnarray}
The second inequality is due to the monogamy relation for the $x$-th power of CREN (1).
\end{proof}

\begin{theorem}\label{TH3}
For the $n$-qubit generalized W-class state $|\psi\rangle\in H_{A_1}\otimes H_{A_2}\otimes...\otimes H_{A_n}$
with $\widetilde{\mathcal{N}}(\rho_{A_1A_{j_i}})\neq0$ for $1\leq i\leq m-1$,    we have
\begin{eqnarray}\label{cay}
\widetilde{\mathcal{N}}_{a}^y(\rho_{A_1|A_{j_1}...A_{j_{m-1}}})<\frac{y}{2^y-1}\sum_{i=1}^{{m-1}}\widetilde{\mathcal{N}}^y_{a}(\rho_{A_1A_{j_i}}),
\end{eqnarray}
where $y\leq0$ and $\rho_{A_1A_{j_1}...A_{j_{m-1}}}$ is the $m$-qubit reduced density matrix as in Theorem 1.
\end{theorem}

\begin{proof}
 Similar to the proof of Theorem1,    for $y\leq0$,    we get
\begin{eqnarray}\nonumber
\widetilde{\mathcal{N}}_{a}^y(\rho_{A_1|A_{j_1}...A_{j_{m-1}}})
&\leq&\frac{y}{2^y-1}\widetilde{\mathcal{N}}^y(\rho_{A_1|A_{j_1}...A_{j_{m-1}}})
\end{eqnarray}

Combining with Lemma 1,    we have
\begin{eqnarray}\nonumber
\widetilde{\mathcal{N}}^y_{a}(\rho_{A_1|A_{j_1}...A_{j_{m-1}}})
&\leq &\frac{y}{2^y-1}\widetilde{\mathcal{N}}^y(\rho_{A_1|A_{j_1}...A_{j_{m-1}}})\\[1mm]\nonumber
&< &\frac{y}{2^y-1}\sum_{i=1}^{m-1}\widetilde{\mathcal{N}}^y(\rho_{A_1A_{j_i}})\\[1mm]\nonumber
&=& \frac{y}{2^y-1}\sum_{i=1}^{m-1}\widetilde{\mathcal{N}}^y_a(\rho_{A_1A_{j_i}}).
\end{eqnarray}
The second inequality is due to the monogamy relation for the $x$-th power of CREN (2).
\end{proof}

As an example,    consider the $5$-qubit generalized $W$-class states (9) with
$a=b_2=\frac{1}{\sqrt{10}}$,    $b_1=\frac{1}{\sqrt{15}}$,    $b_3=\sqrt{\frac{2}{15}}$,    $b_4=\sqrt{\frac{3}{5}}$.
We have

$$\widetilde{\mathcal{N}}_{a}(\rho_{A_1|A_{2}A_{3}})\geq\sqrt[x]{\frac{x}{2^x-1}}\frac{2}{\sqrt{15}}\sqrt[x]{(\frac{1}{\sqrt{10}})^x
+(\sqrt{\frac{2}{15}})^x}
$$
and
$$\widetilde{\mathcal{N}}_{a}(\rho_{A_1|A_{2}A_{3}A_{4}})
\geq\sqrt[x]{\frac{x}{2^x-1}}\frac{2}{\sqrt{15}}\sqrt[x]{(\frac{1}{\sqrt{10}})^x
+(\sqrt{\frac{2}{15}})^x+(\sqrt{\frac{3}{5}})^x}
$$
with $x\geq2$. The optimal lower bounds can be obtained by varying the parameter $x$,
see Fig. 1.

\begin{figure}[h]
\normalsize
\renewcommand{\figurename}{Fig.}
\centering
\includegraphics[width=0.5\textwidth,    height=0.35\textwidth]{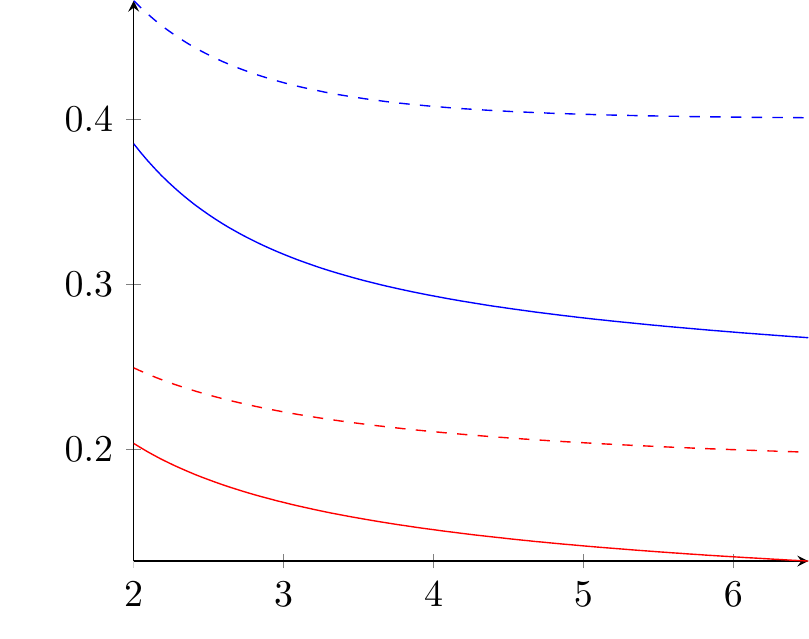}
\caption{{\small solid red line is the lower bound of $\widetilde{\mathcal{N}}_{a}(\rho_{A_1|A_2A_3})$ and solid blue line is the lower bound of $\widetilde{\mathcal{N}}_{a}(\rho_{A_1|A_2A_3A_4})$ as functions of $x\geq2$ from our result,   red dashed line is the lower bound of $C_a(\rho_{A_1|A_2A_3})$ and blue dashed line is the lower bound of $C_a(\rho_{A_1|A_2A_3A_4})$ as functions of $x\geq2$ from \cite{ZXN2}.}}
\label{Fig.1}
\end{figure}

From Fig.1,   one gets that the
optimal lower bounds of $\widetilde{\mathcal{N}}_{a}(\rho_{A_1|A_2A_3})$ and $\widetilde{\mathcal{N}}_{a}(\rho_{A_1|A_2A_3A_4})$ are $0.203$ and $0.385$,     respectively,     attained at $x=2$ while the lower bounds of each in terms of CoA are given by $0.249$ and $0.471$ \cite{ZXN2}. One can see that choosing CRENoA as a mathematical characterization of MoE is better than choosing CoA for $x\geq2$.

\section{Monogamy of R\'{e}nyi-$\alpha$ entanglement}

R\'{e}nyi-$\alpha$ entanglement (R$\alpha$E) is well-defined entanglement measure which is the generalization of entanglement of formation.
For a bipartite pure state $\left| \psi  \right\rangle _{AB}$,     the R$\alpha$E is defined as
\beq\label{q3}
E_{\alpha}(\left| \psi  \right\rangle _{AB}):= S_\alpha(\rho_A) = \frac{1}{1-\alpha}\log _2(\mbox{tr}\rho _A^\alpha)
\eeq
where the R\'{e}nyi-$\alpha$ entropy is $S_\alpha(\rho_A)=[\log _2(\sum_i\lambda_i^{\alpha})]/(1-\alpha)$ with
$\alpha$ being a nonnegative real number and $\lambda_i$ being the eigenvalue of reduced density matrix
$\rho_A$.
For a bipartite mixed state $\rho _{AB}$,     the R$\alpha$E
is defined via the convex-roof extension
\beqa\label{q4}
E_\alpha(\rho _{AB})=\mbox{min} \sum_i p_i E_\alpha(\ket{\psi _i }_{AB})
\eeqa
where the minimum is taken over all possible pure state decompositions of ${\rho_{AB}=\sum\limits_i{p_i
\left| {\psi _i } \right\rangle _{AB} \left\langle {\psi _i } \right|} }$. In particular,     for a
two-qubit mixed state,     the R$\alpha$E with $\alpha\geq 1$ has an analytical formula which is expressed
as a function of the SC \cite{JSK3}
\beqa\label{q5}
E_\alpha \left( {\rho _{AB} } \right) = f_\alpha \left[ {C^2 \left( {\rho _{AB} } \right)} \right]
\eeqa
where the function $f_\alpha \left( x \right)$ has the form
\beq\label{q6}
f_\alpha \! \left( x \right)\!= \!\frac{1}{{1 - \alpha }}\!\log _2 \!\left[ {\left( {\frac{{1 \!-\!
\sqrt {1 - x} }}{2}} \right)^\alpha  \!\!\!\! +\! \left( {\frac{{1 \!+\! \sqrt {1 - x} }}{2}}
\right)^\alpha  } \right].
\eeq
Recently,     Wang \emph{et al} further proved that the formula in (14) holds for the order
$\alpha\geq(\sqrt{7}-1)/2\simeq 0.823$ \cite{YXW}.

From Theorem 2 in Ref. \cite{WS},    one has that for a bipartite $2 \otimes d$ mixed state $\rho _{AC}$,     the
R\'{e}nyi-$\alpha$ entanglement has an analytical expression
\beqa\label{q12}
E_\alpha  \left( {\rho _{AC} } \right) = f_\alpha  \left[ {C^2 \left( {\rho _{AC} } \right)} \right]
\eeqa
where the order $\alpha$ ranges in the region $[(\sqrt 7  - 1)/2,    (\sqrt {13}  - 1)/2]$.

\begin{theorem}\label{E}
For the $n$-qubit generalized W-class states $|\psi\rangle\in H_{A_1}\otimes H_{A_2}\otimes...\otimes H_{A_n}$,     we have
\begin{eqnarray}\label{2}
E_{\alpha}(|\psi\rangle_{A_1|A_2...A_{n}})\leq\sum_{i=2}^{{n}}E_{\alpha}(\rho_{A_1A_{i}}),
\end{eqnarray}
where $\rho_{A_1A_i}$,     $2\leq i\leq n$,     is the $2$-qubit reduced density matrix of $|\psi\rangle$ and the order $\alpha$ ranges in the region $[(\sqrt 7  - 1)/2,    (\sqrt {13}  - 1)/2]$.
\end{theorem}

\begin{proof}
For the $n$-qubit generalized W-class states $|\psi\rangle$,     we have
\begin{eqnarray}\nonumber
E_{\alpha}(|\psi\rangle_{A_1|A_2...A_{n}})
&=&f_\alpha\left(C^2(|\psi\rangle_{A_1|A_2...A_{n}})\right)\\[1mm]\nonumber
&=&f_\alpha(\sum_{i=2}^{n}C^2(\rho_{A_1A_{i}}))\\[1mm]\nonumber
&\leq& \sum_{i=2}^{n}f_\alpha(C^2(\rho_{A_1A_{i}}))\\[1mm]\nonumber
&=& \sum_{i=2}^{n}E_\alpha(\rho_{A_1A_{i}}),
\end{eqnarray}
where $f_\alpha \! \left( x \right)\!= \!\frac{1}{{1 - \alpha }}\!\log _2 \!\left[ {\left( {\frac{{1 \!-\!
\sqrt {1 - x} }}{2}} \right)^\alpha  \!\!\!\! +\! \left( {\frac{{1 \!+\! \sqrt {1 - x} }}{2}}
\right)^\alpha  } \right]$.
We have used in the first and last equalities that the entanglement
of formation obeys the relation (19). The second equality is due to the fact that
$C^2(|\psi\rangle_{A_1...A_{n}})=\sum_{i=2}^{n}C^2(\rho_{A_1A_{i}}).$
The inequality is due to the fact that the R\'{e}nyi-$\alpha$ entanglement $E_\alpha \left( {C^2 } \right)$ with
$\alpha\in[(\sqrt7-1)/2,    (\sqrt{13}-1)/2]$ is monotonic increasing and concave as a function of the squared
concurrence ${C^2 }$ \cite{WS}.
\end{proof}

 Next we will present an upper bound of SR$\alpha$E. Before giving the result,   we consider the following lemma.
\begin{lemma}\label{c=ca}\cite{WS}
Let $\psi_{A_1\cdots A_n}$ be a generalized W-class state in (11).
For any $m$-qubit subsystems $A_1A_{j_1}\cdots A_{j_{m-1}}$ of $A_1\cdots A_n$ with $2 \leq m \leq  n-1$,
the reduced density matrix $\rho_{A_1A_{j_1}\cdots A_{j_{m-1}}}$ of $\psi_{A_1\cdots A_n}$ is a mixture of a $m$-qubit generalized W-class state
and vacuum.
\end{lemma}

In the following,  we assume the $m$-qubit subsystems $A_1A_{j_1}\cdots A_{j_{m-1}}$ of $A_1\cdots A_n$ with $2 \leq m \leq  n-1$ is exactly $A_1\cdots A_n$. Then we can have the result below.

\begin{theorem}\label{TH4}
For the $n$-qubit generalized W-class states $|\psi\rangle\in H_{A_1}\otimes H_{A_2}\otimes...\otimes H_{A_n}$,     we have
\begin{eqnarray}\label{2}
E_{\alpha}^2(\rho_{A_1|A_2...A_{n}})\leq(n-1)\sum_{i=2}^{{n}}E_{\alpha}^2(\rho_{A_1A_{i}}),
\end{eqnarray}
where $E^2_\alpha(\rho_{A_1|A_2\dots A_n})$ quantifies the entanglement in the partition
$A_1|A_2\dots A_n$ and $E_\alpha^2(\rho_{A_1A_i})$ quantifies the one in two-qubit subsystem $A_1A_i$
with the order $\alpha\in[(\sqrt7-1)/2,    (\sqrt{13}-1)/2]$.
\end{theorem}

\begin{proof}
We first consider the monogamy relation in an $n$-qubit pure state
$\psi_{A_1A_2\dots A_n}$. Thus we can obtain
\begin{eqnarray}\nonumber
E_{\alpha}^2(|\psi\rangle_{A_1|A_2...A_{n}})
&\leq&(\sum_{i=2}^{{n}}E_{\alpha}(\rho_{A_1A_{i}}))^2\\[1mm]\nonumber
&\leq& (\sum_{i=2}^{{n}}1^2)(\sum_{i=2}^{{n}}E_{\alpha}^2(\rho_{A_1A_{i}}))\\[1mm]\nonumber
&=& (n-1)(\sum_{i=2}^{{n}}E_{\alpha}^2(\rho_{A_1A_{i}})),
\end{eqnarray}
where in the first inequality we have used Theorem 3 and $a^2\leq b^2$ for $0\leq a\leq b$,  and in the second inequality we have used the Cauchy-Schwarz inequality.

Next from Lemma2,   we consider $\rho_{A_1\cdots A_n}$ is a mixture of a $n$-qubit generalized W-class state
and vacuum. Then since we have the pure decomposition of $\rho_{A_1\cdots A_n}$,
\begin{eqnarray}\nonumber
\rho_{A_1A_2\dots A_n}=\sum_j p_j \ket{\psi_j}_{A_1A_2\dots A_n}\bra{\psi_j},
\end{eqnarray}
Thus,   we can obtain
\begin{eqnarray}\nonumber
E_{\alpha}^2(\rho_{A_1|A_2\dots A_n})
&=&[\sum_j p_j E_\alpha(\ket{\psi_j}_{A_1|A_2\dots A_n})]^2\\[1mm]\nonumber
&\leq&[\sum_j p_j (\sum_{i=2}^{{n}}E_{\alpha}(\rho_{A_1A_{i}}))]^2\\[1mm]\nonumber
&=&[\sum_{i=2}^{{n}}(\sum_j p_j E_{\alpha}(\rho_{A_1A_{i}}))]^2\\[1mm]\nonumber
&\leq& (\sum_{i=2}^{{n}}1^2)[\sum_{i=2}^{{n}}(\sum_j p_j E_{\alpha}(\rho_{A_1A_{i}}))^2]\\[1mm]\nonumber
&=& (n-1)(\sum_{i=2}^{{n}}E_{\alpha}^2(\rho_{A_1A_{i}})),
\end{eqnarray}
where in the first inequality we have used Theorem 3 and $a^2\leq b^2$ for $0\leq a\leq b$ ,  and in the second inequality we have used the Cauchy-Schwarz inequality. The last equality is due to $\sum_j p_j=1$.
\end{proof}

As an example,   we still consider the $5$-qubit generalized $W$-class states (9) with
$a=b_2=\frac{1}{\sqrt{10}}$,    $b_1=\frac{1}{\sqrt{15}}$,    $b_3=\sqrt{\frac{2}{15}}$,    $b_4=\sqrt{\frac{3}{5}}$.
We have
$$E_{\alpha}^2(\rho_{A_1|A_{2}A_{3}})\leq2\left(E_{\alpha}^2(\rho_{A_1|A_{2}})+E_{\alpha}^2(\rho_{A_1|A_{3}})\right)
$$
and
$$E_{\alpha}^2(\rho_{A_1|A_{2}A_{3}A_{4}})\leq3\left(E_{\alpha}^2(\rho_{A_1|A_{2}})+E_{\alpha}^2(\rho_{A_1|A_{3}})+E_{\alpha}^2(\rho_{A_1|A_{4}})\right)
$$
where
$$E_{\alpha}(\rho_{A_1|A_{2}})=\!\frac{1}{{1 - \alpha }}\!\log _2 \!\left[ {\left( {\frac{{1 \!-
\sqrt {\frac{73}{75}} }}{2}} \right)^\alpha  \!\!\!\! +\! \left( {\frac{{1 \!+\sqrt { \frac{73}{75}} }}{2}}
\right)^\alpha  } \right]
$$

$$E_{\alpha}(\rho_{A_1|A_{3}})=\!\frac{1}{{1 - \alpha }}\!\log _2 \!\left[ {\left( {\frac{{1 \!-\!
\sqrt {\frac{217}{225}} }}{2}} \right)^\alpha  \!\!\!\! +\! \left( {\frac{{1 \!+\! \sqrt {\frac{217}{225}} }}{2}}
\right)^\alpha  } \right]
$$
and
$$E_{\alpha}(\rho_{A_1|A_{4}})=\!\frac{1}{{1 - \alpha }}\!\log _2 \!\left[ {\left( {\frac{{1 \!-\!
\sqrt {\frac{63}{75}} }}{2}} \right)^\alpha  \!\!\!\! +\! \left( {\frac{{1 \!+\! \sqrt {\frac{63}{75}} }}{2}}
\right)^\alpha  } \right]
$$
with the order $\alpha\in[(\sqrt7-1)/2,    (\sqrt{13}-1)/2]$. See Fig. 2.
\begin{figure}[h]
\normalsize
\renewcommand{\figurename}{Fig.}
\centering
\includegraphics[width=0.45\textwidth,  height=0.35\textwidth]{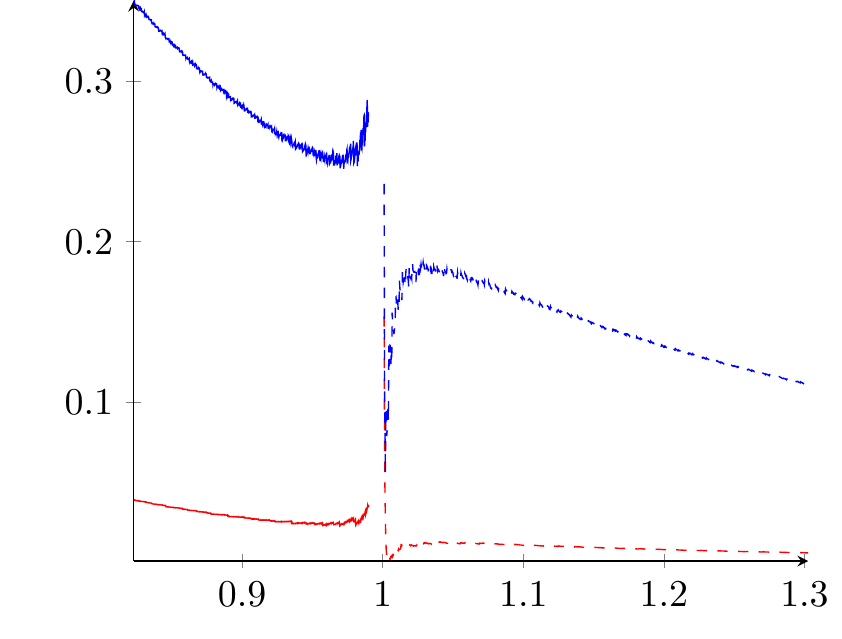}
\caption{{\small solid red line is the upper bound of $E_{\alpha}^2(\rho_{A_1|A_{2}A_{3}})$ and solid blue line is the upper bound of $E_{\alpha}^2(\rho_{A_1|A_{2}A_{3}A_{4}})$ as functions of $\alpha$ when $\alpha$ ranges in the region $[(\sqrt 7  - 1)/2,   0.99]$. When $\alpha$ ranges in the region $[1.001,   (\sqrt 13  - 1)/2]$,   as the red dashed line and the blue dashed line show,   we do not have an upper bound of $E_{\alpha}^2(\rho_{A_1|A_{2}A_{3}})$ and $E_{\alpha}^2(\rho_{A_1|A_{2}A_{3}A_{4}})$ in this example.}}
\label{Fig.2}
\end{figure}

From Fig.2,   one gets that the optimal upper
bounds of $E_{\alpha}^2(\rho_{A_1|A_{2}A_{3}})$ and $E_{\alpha}^2(\rho_{A_1|A_{2}A_{3}A_{4}})$ are $0.02334$ and $0.24211$ attained at $\alpha=0.971$ when $\alpha\in[(\sqrt7-1)/2,    0.99]$. This upper bounds can be easily generalized to arbitrary $n$-qubit generalized W-class states $|\psi\rangle\in H_{A_1}\otimes H_{A_2}\otimes...\otimes H_{A_n}$.

\section{Conclusions and remarks}

We have investigated the monogamy relations of muliti-qubit generalized W-class states in terms of CRENoA and SR$\alpha$E. We have proved that the monogamy inequality of $x$-th power for CRENoA when $x\geq2$ and $x\leq0$. Our result shows that choosing CRENoA as a mathematical characterization of the monogamy of entanglement is better than choosing CoA for $x\geq2$. We also show the monogamy inequality for SR$\alpha$E when $\alpha$ ranges in the region $[(\sqrt 7  - 1)/2,    (\sqrt {13}  - 1)/2]$. We can find the optimal upper bound for $E_{\alpha}^2(\rho_{A_1|A_2...A_{n}})$ when the order $\alpha\in[(\sqrt7-1)/2,    (\sqrt{13}-1)/2]$ by using our approach in Theorem 4.
It is still an open problem to be answered that whether there exists the monogamy inequality for SR$\alpha$E when $\alpha\geq(\sqrt {13}  - 1)/2) $ in generalized W-class states.

\section*{Acknowledgments}
 This work is supported by the NSFC 11571119 and NSFC 11475178.


\begin{thebibliography}{99}
\bibitem{CKW} V. Coffman,  J. Kundu and W. K. Wootters,  Phys. Rev. A \textbf{ 61},  052306 (2000).

\bibitem{OV} T. Osborne and F. Verstraete,  Phys. Rev. Lett. \textbf{ 96},  220503 (2006).

\bibitem{KW} M. Koashi and A. Winter,  Phys. Rev. A \textbf{69},  022309 (2004).

\bibitem{t1} F. Mintert,   M. Ku\'{s},   and A. Buchleitner,   Phys. Rev. Lett. \textbf{92},   167902 (2004).

\bibitem{t2} K. Chen,   S. Albeverio,   and S. M. Fei,   Phys. Rev. Lett. \textbf{95},   040504 (2005).

\bibitem{t3} H. P. Breuer,   J. Phys. A: Math. Gen. \textbf{39},   11847 (2006).

\bibitem{t4} H. P. Breuer,   Phys. Rev. Lett. \textbf{97},   080501 (2006).

\bibitem{t5} J. I. de Vicente,   Phys. Rev. A \textbf{75},   052320 (2007).

\bibitem{t6} C. J. Zhang,   Y. S. Zhang,   S. Zhang,   and G. C. Guo,   Phys. Rev. A \textbf{76},   012334 (2007).

\bibitem{JMR} J. M. Renes and M. Grassl,  Phys. Rev. A \textbf{74},  022317 (2006).

\bibitem{WKW} W. K. Wootters,  Phys. Rev. Lett. \textbf{80},  2245 (1998).

\bibitem{TJO} T. J. Osborne,   and F. Verstraete,   Phys. Rev. Lett. \textbf{96},   220503 (2006).

\bibitem{YKB} Y. K. Bai,   M. Y. Ye,   and Z. D. Wang,   Phys. Rev. A \textbf{80},   044301(2009).

\bibitem{ZXN1} X. N. Zhu,   and S. M. Fei,   Physical Review A \textbf{90},   024304 (2014)

\bibitem{JSK1} J. S. Kim,   Phys. Rev. A \textbf{90},    062306 (2014).

\bibitem{SL} S. Lee,   D. P. Chi,   S. D. Oh,   and J. Kim,   Phys. Rev. A \textbf{68},   062304 (2003).

\bibitem{RH}R. Horodecki,   P. Horodecki,   and M. Horodecki,   Phys. Rev. A \textbf{210},   377 (1996).

\bibitem{YL} Y. Luo,   Y. M. Li,   Phys. Rev. A \textbf{362},   511-520 (2015).

\bibitem{WS} W.Song,   Y. K. Bai,   M.Yang,   M.Yang and Z.L. Cao,   Phys. Rev. A \textbf{93},    022306  (2015).

\bibitem{PR} P. Rungta,   V. Bu$\check{\text{z}}$ek,   C. M. Caves,   M. Hillery,   and G. J. Milburn,   Phys. Rev. A \textbf{64},   042315 (2001).

\bibitem{CSY} C. S. Yu,   and H. S. Song,   Phys. Rev. A \textbf{77},    032329 (2008).

\bibitem{JSK2} J. S. Kim,   A. Das and B. C. Sanders,   Phys. Rev. A \textbf{79},   012329 (2008).

\bibitem{ZXN2} X. N. Zhu,   S. M. Fei,   Quant. Inf. Process \textbf{16},   53  (2017).

\bibitem{JSK3}J. S. Kim,    and B. C. Sanders,   Phys. Rev. A \textbf{43},   445305 (2010).

\bibitem{YXW}Y. X. Wang,   L. Z. Mu,   and V.Vedral,   and H.Fen,   Phys. Rev. A \textbf{93},   022324 (2016).

\end{thebibliography}
\end{document}